%% file: convehRevised.tex
\documentclass[5p,times]{elsarticle}

\usepackage{amsmath}
\usepackage{amsthm}
\usepackage{amssymb}
\usepackage{graphics} 
\usepackage{epstopdf} 
\usepackage{epsfig} 
\usepackage{color}
\usepackage{calc}
\usepackage[ruled,vlined]{algorithm2e}
\SetKwFor{For}{for}{}{end}
\SetKwFor{While}{while}{}{end}

\usepackage{mathtools}
\DeclarePairedDelimiter{\abs}{\lvert}{\rvert}

\newtheorem{theorem}{Theorem}
\newtheorem{definition}{Definition}
\newtheorem{lemma}{Lemma}
\newtheorem{proposition}{Proposition}

\newtheorem{problem}{Problem}

\newtheorem{corollary}{Corollary}

\usepackage{lineno,hyperref}

\journal{European Journal of Control}

\begin{document}

\begin{frontmatter}

\title{Resilient Coordinated Movement of Connected Autonomous Vehicles}


\author[safi]{Mostafa Safi}
\corref{mycorrespondingauthor}
\cortext[mycorrespondingauthor]{Corresponding author}
\ead{halebi@aut.ac.ir}

\author[dibaji]{Seyed Mehran Dibaji}
\ead{dibaji@mit.edu}

\author[pirani]{Mohammad Pirani}
\ead{pirani@kth.se}

\address[safi]{Aerospace Engineering Department, Amirkabir University of Technology, Tehran, Iran}
\address[dibaji]{Mechanical Engineering Department, Massachusetts Institute of Technology, Cambridge, MA, USA}
\address[pirani]{University of Toronto, Toronto, Canada}

\begin{abstract}
In this paper, we consider coordinated movement of a network of vehicles consisting of a bounded number of \textit{malicious} agents, that is, vehicles must reach consensus in longitudinal position and a common predefined velocity. The motions of vehicles are modeled by double-integrator dynamics and communications over the network are asynchronous with delays. Each normal vehicle updates its states by  utilizing the information it receives from vehicles in its vicinity. On the other hand, misbehaving vehicles make updates arbitrarily and might threaten the consensus within the network by intentionally changing their moving direction or broadcasting faulty information in their neighborhood. We propose an asynchronous updating strategy for normal vehicles, based on filtering extreme values received from neighboring vehicles, to save them from being misguided by malicious vehicles. We show that there exist topological constraints on the network in terms of graph robustness under which the vehicles resiliently achieve coordinated movement. Numerical simulations are provided to evaluate the results.
\end{abstract}


\end{frontmatter}


\section{Introduction}
Modern intelligent vehicles are not only used for driving but are processors that can perform complicated tasks and connect to their surroundings \cite{Shensherman}. The advent of ever-growing Internet of Vehicles, along with cloud services, enables vehicles to communicate important information which can potentially be used for management of the networked vehicles or increasing the reliability of each vehicle's estimation and control individually \cite{BartBesselink2, LiangJohansson}. However, as the vehicles become more connected, they become more prone to adversarial actions and cyber-attacks. To this end, devising defense mechanisms, to increase the security for both intra-vehicle networking and inter-vehicular communications, is of great importance \cite{Jia, alipour2020impact, wired, Koscher}. An efficient defense mechanism must be able to prevent the attack as much as possible, detect the attack in case of happening, and satisfy a level of resilience in performing tasks despite the existence of an attack \cite{dibaji2019systems}. The focus of the current paper is the application of resilience methods to a Cooperative Adaptive Cruise Control (CACC) strategy (Fig.~\ref{fig.CACC}).

\begin{figure}[t] 
	\begin{center}
	\includegraphics[scale=0.45]{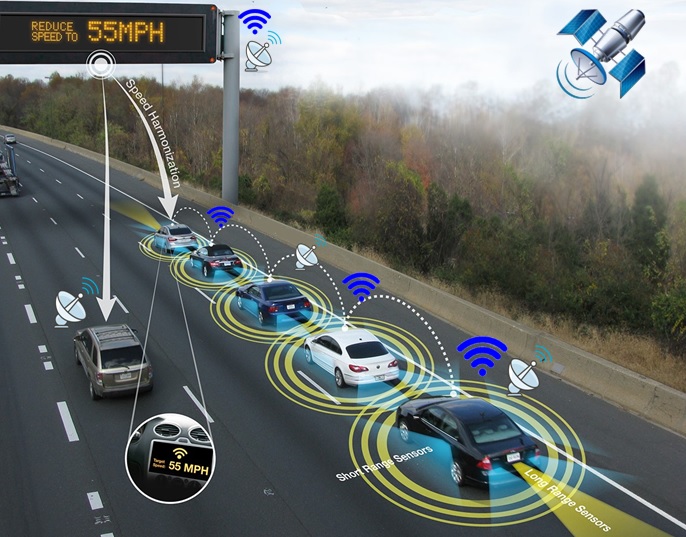}
	\caption{Cooperative adaptive cruise control in a network of vehicles with active and/or passive sensing systems, photo source (with some edits): US Department of Transportation.}
	\label{fig.CACC}
	\end{center}
\end{figure} 

Vehicle-to-vehicle (V2V) communications can provide direct data transfer which possess a much lower delay compared to radars \cite{van2006impact} and enable vehicles to move close together, while collision avoidance algorithms and congestion control protocols assist this strategy \cite{yuan2020race,forghani2015safety,ahmad2019v2v}. This will increase the amount of road throughput and reduce the need for developing more road networks. Cooperative adaptive cruise control, as one of the applications of V2V communications, is among the widely used methods in controlling highway traffic systems \cite{BartBesselink}. In this approach, vehicles tend to follow specific speeds while maintaining a safe distance from each other and at the same time consume as little space as possible in the highway to facilitate the traffic flow \cite{Papageorgiou}. Since wireless communication plays pivotal roles in CACC, we must make it resilient to malicious actions \cite{amoozadeh}. Attack-resilient algorithms in CACC context refer to a class of actions taken to bypass the attacker or mitigate its effects in order to improve performance in vehicle formation and velocity tracking. Similarly, in multi-agent systems, various consensus setups have been widely studied in the past decade \cite{mesbahi,renbook2}, where locally connected agents interact to achieve agreement by reaching a common state. In this literature, \textit{resilient consensus} expresses the situation where some of the agents in the network take some actions to deceive the others, are possibly crashed, or intentionally evade executing the local state updating rule. These types of consensus problems are frequently studied in the computer science field by designing distributed control algorithms(see, e.g., \cite{lynch}).

There are various techniques to relieve or counteract the effects of cyber attacks in multi-agent systems. In some solutions, each agent identifies the adversarial agents of the network by observing their information history -- a sort of fault detection and isolation strategy \cite{pasq,shames,sundaram}. 
However, usually in these techniques, each agent needs global knowledge of the network such as topology, which is neither desirable in distributed algorithms nor scalable. 
It is shown that to overcome the misbehaving of $f$ malicious vehicles, the topology is required to be at least $(2f + 1)$-connected \cite{pasq,sundaram}.


In another type of distributed algorithms for resilient consensus, each agent utilizes a kind of filtering of the information packets delivered from the agents in its neighborhood containing extreme or invalid values at each time step. This class of algorithms are often called Mean Subsequence Reduced (MSR) algorithms, which was firstly introduced in \cite{Azadmanesh1993} and have been extensively used in the literature of computer science \cite{Azadmanesh2002,Azevedo,zohir,lynch,plunkett,Vaidya} 
as well as estimation and control \cite{safi2020filtering,dibajiishii2015,dibajiishiiSCL2015,LeBlancPaper,zhang1}.  \textcolor{black}{With MSR-type algorithms, there is no need to know the entire topology of the network; instead, a parameter $f$ is considered to secure the consensus in the worst case scenario, in which $f$ is an upper bound for the number of malicious nodes. That means there is no need that the agents know $f$ accurately. In fact, each node considers that there are at most $f$ malicious nodes within the neighborhood (or within the entire network). These types of assumptions are common in the literature of robust control.} In this literature, convergence of these algorithms is guaranteed by some constraints on the topology. While the traditional connectivity constraints are not enough for convergence of these algorithms, as is stated in \cite{zhang1}, graph \textit{robustness} has been recently used as a successful connectivity measure for different consensus problems to ensure that a network achieves agreement \cite{dibaji2017resilient,dibaji2019resilient}.

Our contributions in this paper are threefold:

\begin{itemize}
\item {\bf Modeling and formulation of resilient longitudinal coordinated movement of autonomous vehicles:} we use resilient consensus notions introduced in computer science literature \cite{dibaji2017resilient} to solve a real world problem in multi-vehicle coordinated movement. We consider longitudinal cooperative cruise control in the presence of some malicious vehicles in the network. Therefore, we assume that the vehicles move in parallel or each vehicle is equipped with a Collision Avoidance System (CAS) to prevent colliding with its neighbors and can switch its lane to overtake frontier vehicles if required. We use the typical second-order dynamics to model each autonomous vehicle within the network. Also, each vehicle makes updates based on its current position and velocity and those most recently received from its neighbors. The control input of each vehicle is applied through its acceleration.

\item { \bf Developing distributed algorithm and update rule for each vehicle to reach agreement considering asynchrony and delays in communications:} We develop our distributed algorithms in an asynchronous setting, where each normal vehicle may decide to update from time to time with possibly delayed data packets received from the agents in its neighborhood. This is a susceptible situation which allows the malicious vehicles to take advantage and broadcast different information, including their states, in different time intervals to reachable vehicles or change their movement quickly to appear in different states in perspective of the other vehicles. We suppose that the worst-case scenario may happen, where the malicious vehicles have global knowledge about the topology, updating times, the transmitted data packets by normal vehicles, and even delays of communications. On the other hand, the normal vehicles only have access to the data directly received from their neighbors; thus they cannot predict adversaries' behavior.

\item {\bf Analyzing the topology constraint required for the resilient cooperative cruise control:} The proposed algorithm and update rule cannot immediately lead to the success of the resilient coordinated movement. There must be specific constraints on the topology to guarantee the convergence. We analyze these topology constraints based on robustness notions developed in the literature.
\end{itemize} 

The outline of this paper is as follows. In Section~\ref{sect:problemsetting}, we introduce the preliminaries and problem setup. Our main results including the update rule, the filtering algorithm, and the required topology constraints are presented in Section~\ref{sect: main}. We show the effectiveness of our strategy by simulation examples in Section~\ref{sect: simulations}. Finally, Section~\ref{sect: conclusion} concludes the paper.

\section{Preliminaries}\label{sect:problemsetting}

\subsection{Graph Theory Notions}\label{graphnot}
According to \cite{mesbahi}, we recall some preliminary concepts on graphs. A digraph with $n$ nodes $(n> 1)$ is defined as $\mathcal{G}[k]=\left(\mathcal{V}[k],\mathcal{E}[k]\right)$, $k \in \mathbb{Z}_+$, with the set of nodes $\mathcal{V}=\{1,\ldots,n\}$ and the set of edges  $\mathcal{E}\subseteq \mathcal{V}\times\mathcal{V}$. We sometimes drop down the step time $k$ for the sake of simplicity. An incoming link from node $j$ to node $i$ is denoted by $(j,i)\in \mathcal{E}$. The adjacency matrix corresponding to $\mathcal{G}[k]$ is defined by $A[k] \in \mathbb{R}^{n\times n}$. A graph is called complete if $\forall i,j \in \mathcal{V}, i \neq j: (i,j) \in \mathcal{E}$. For node $i$, the set of its neighbors is denoted by $\mathcal{N}_i=\{j \vert (j,i)\in \mathcal{E}\}$ and the number of its neighbors, i.e. its degree, is denoted by ${d}_i=\abs{\mathcal{N}_i}$.

A path is a subset of nodes $\mathcal{P} = \{i \vert (i+1,i)\in \mathcal{E},i = 1, \ldots,p-1, p>1\}$. If there is a path between each pair of nodes in the directed graph $\mathcal{G}$, it is said to be strongly connected.
The vertex connectivity $K(\mathcal{G})$ of the graph 
$\mathcal{G}$ is the minimum number of nodes such that by removing them and all associated edges, the graph is not strongly connected anymore. Then, the graph is said to be $\kappa$-connected if $K(\mathcal{G}) \geq \kappa$. A directed graph is said to have a directed spanning tree if there exists a node in the graph from which there is a path to every other node. Note that we will use the terms node and vehicle interchangeably in this paper.

Among connectivity measures of a graph, robustness is the critical topological notion for the MSR-type algorithms to achieve resilient consensus. Robust graphs were coined in \cite{zhang1} 
for the analysis of resiliency of consensus in multi-agent systems  against cyber-attacks. 

\begin{definition}   \rm \label{(def:robustgraph)}
($(r,s)$-robust) The digraph $\mathcal{G}= (\mathcal{V},\mathcal{E})$ is $(r,s)$-robust $(r,s<n)$ if for every pair of nonempty disjoint subsets $\mathcal{V}_1,\mathcal{V}_2 \subset \mathcal{V}$, at least one of the following conditions is satisfied:
\begin{enumerate}
\item[1.] $\mathcal{X}_{\mathcal{V}_1}^r =\mathcal{V}_1$,
\item[2.] $\mathcal{X}_{\mathcal{V}_2}^r=\mathcal{V}_2$,
\item[3.] $\abs{\mathcal{X}_{\mathcal{V}_1}^r} +\abs{\mathcal{X}_{\mathcal{V}_2}^r} \geq s$,
\end{enumerate}
where $\mathcal{X}^r_{\mathcal{V}_{\ell}}$ is the set of all nodes in ${\mathcal{V}_{\ell}}$ which have at least $r$ incoming edges from outside of ${\mathcal{V}_{\ell}}$.
In particular, graphs which are $(r,1)$-robust are called $r$-robust.
\end{definition}
To have a better understanding of the notions of robustness, refer to \ref{sect: appendix}.
\subsection{Physics of the System}

Consider a network of vehicles driving in a road whose communications are through the directed graph $\mathcal{G}$. 
Each vehicle $i\in\mathcal{V}$ has a second-order dynamic system given by
\begin{align}
\begin{split}
\dot{x_i}(t)=v_{i}(t),~~ \dot{v}_i(t)=u_i(t),\qquad {}  i=1,\ldots,n, 
\end{split}
\label{eq.dyn}
\end{align}
where $x_i(t)\in{\mathbb{R}}$ and $v_i(t)\in{\mathbb{R}}$ are the longitudinal position and velocity of the $i$th vehicle, respectively, and $u_i(t)$ is its control input at time $t\geq 0$. 
The discrete form of the dynamic system \eqref{eq.dyn} 
with sampling period $T$ is represented as
\begin{align}
\begin{split}
x_i[k+1]&=x_i[k]+Tv_i[k]+\frac{T^2}{2}u_i[k],\\
v_i[k+1]&=v_i[k]+Tu_i[k],\qquad  i=1,\ldots,n,
\end{split}
\label{eq.dyndiscrete}
\end{align}
where $x_i[k]$, $v_i[k]$, and $u_i[k]$ are, respectively, the position, the velocity, and the control input of the $i$th vehicle at $t=kT$ for $k\in \mathbb{Z}_+$, where $T$ is the sampling period \cite{renbook2}. 

\subsection{Problem Setup}
In this paper, we investigate the coordinated movement of networked vehicles in the sense that they reach a same fixed velocity asymptotically leading to a formation with a predefined (safe) inter-vehicular distance: $x_i[k]-x_j[k] \rightarrow \delta_{ij}$, $v_i[k]\rightarrow r$ as $k\rightarrow\infty$, $\forall i,j\in \mathcal{V}$, where $\delta_{ij}$ is \textcolor{black}{the desired relative position of node $i$ with respect to $j$} and $r$ is the desired velocity of the networked vehicles which is assumed to be known for all. \textcolor{black}{We intentionally defined $\delta_{ij}$ as a distributed parameter because of two reasons: i) we aimed to minimum the global information known by the agents in the network, ii) it gives the vehicles the opportunity to set the inter-vehicular distances based on their size and sensing systems. For example, it would be safer for two trucks to have a longer relative distance.}

In this work, we investigate the case where some vehicles misbehave because of damage, disturbances, or various cyber attacks. \textcolor{black}{Note that we focus on the consequences of such attacks or failures in consensus of the network, not the source of attacks (or failures). For example, some malicious vehicles might intentionally send false data to their neighbors in the network or, alternatively, some vehicles might suddenly crash and lose their control. In other words, nature of failures makes no difference in our problem setup.} In order to formulate the problem, we elaborate on some notions regarding the communications in the network and consensus in the presence of malicious vehicles.

There are two possible situations in which malicious vehicles might deceive the normal vehicles and prevent them from reaching consensus based on their equipment (refer to Fig.~\ref{fig.CACC}), described below:

\begin{itemize}
\item[i)] Vehicles use an active sensing system: they can estimate the state of their neighbors on their own (for example using their 3D camera or LiDAR system.
\item[ii)] Vehicles use a passive sensing system: they trust the information they receive from their neighbors (for example using their GPS receiver or wireless communication system which are susceptible to cyber attacks).
\end{itemize}

In the first case, malicious vehicles can intentionally change their moving direction or oscillate by avoiding any prescribed update rule and choosing arbitrary control inputs. Accordingly, all normal vehicles are supposed to follow. In the second case, malicious vehicles can arbitrarily broadcast any information in their neighborhood to deceive their neighbors. Note that the dynamics for all vehicles still remain the same as \eqref{eq.dyndiscrete}.

Accordingly, we divide the vehicles into two groups of malicious and normal vehicles as follows.

\begin{definition}  \label{def.mal} \rm
(Malicious and Normal Vehicles) Vehicle $i$ is called malicious if it can evade following any prescribed algorithm for updating its control input or broadcast false state feedback to its neighbors. The malicious vehicles are assumed to be omniscient, i.e. they have full knowledge of the topology, updating times, the transmitted data packets by normal vehicles, and even delays for all communication links and all $k \geq 0$ --  this is a reasonable assumption as it takes the worst case scenarios. Otherwise, it is called normal. The set of malicious vehicles is denoted by $\mathcal{M}\subset\mathcal{V}$. 
\end{definition}

Furthermore, we assume that the number of malicious vehicles, all over the network or at least in the neighborhood of each normal vehicle, is upper bounded. 

\begin{definition} \rm
($f$-total Malicious Model) The network is $f$-total malicious if the number of malicious vehicles in the whole network is at most $f$.
\end{definition} 

\begin{definition} \rm
($f$-local Malicious Model) The network is $f$-local malicious if the number of malicious vehicles in the neighborhood of each normal vehicle $i$ is bounded by $f$, i.e., $\abs{\mathcal{N}_i \cap \mathcal{M}} \leq f$, $i \in \mathcal{V}/\mathcal{M}$.  
\end{definition} 

Now, we formally define the concept of resilient coordinated movement for the proposed network of vehicles as follows.

\begin{definition} \label{def.velcons}\rm
(Resilient Coordinated Movement) For any possible set of malicious vehicles and their arbitrarily chosen inputs, the network of normal vehicles is said to achieve resilient coordinated movement if it holds that 
$x_j[k]-x_i[k] \rightarrow \delta_{ij}$, $v_i[k]\rightarrow r$ as $k\rightarrow\infty$, $\forall i,j\in \mathcal{V} \setminus \mathcal{M}$, where $\delta_{ij}$ is a predefined (safe) inter-vehicular distance between the nodes $i$ and $j$ and $r$ is the desired velocity of the networked normal vehicles.
\end{definition}

In practice, the parameter $\delta_{ij}$ can be properly set by the CAS of each vehicle to keep a safe distance with the other vehicles in the ego lane avoiding any collision. Furthermore, the vehicles might not have synchronous and delay-free communications with all the neighbors simultaneously. Thus, the solution must be robust against both delays and asynchrony which are very important in real world applications. 

Finally, the main problem which we consider in this paper is as follows: 
\begin{problem} \rm
Under the $f$-total / $f$-local malicious model, find a condition on the network topology so that the normal vehicles reach the resilient coordinated movement using an asynchronous update rule. 
\end{problem}

\section{Main Results} \label{sect: main}
In this section, we propose the update rule and MSR-type algorithm by which the normal vehicles are able to reach the coordinated movement in the presence of misbehaving vehicles. We present the updating strategy compatible with communication delays and asynchrony. Therefore, the vehicles are allowed to update occasionally using delayed data packets. Note that the updating strategies that are developed here must be enhanced with the constraints on the topology of the network which will be discussed in the next section.

\subsection{Update Rule}
We modified the algorithm and update rule proposed in \cite{dibajiishiiSCL2015} to ADP-MSR (Asynchronous Double-integrator Position-based Mean Subsequence Reduced), which suits the problem of resilient coordinated movement. Each normal vehicle distributively uses the relative position to its neighbors and its own velocity as the feedback. 

To develop the update rule for coordinated movement, first we use a change of variables as follows:

\begin{equation} \label{eq.varch}
\begin{split}
x_i[k]&=p_i[k]+kTr\\
v_i[k]&=q_i[k]+r,
\end{split}
\end{equation}
where, $p_i[k]$ and $q_i[k]$ are the transformed variables and $r$ is the desired velocity of the networked vehicles based on Def.~\ref{def.velcons}. Substituting the new variables into the dynamic system \eqref{eq.dyndiscrete}, we have:
\begin{align} \label{eq.mdyn}
\begin{split}
p_i[k+1]&=p_i[k]+Tq_i[k]+\frac{T^2}{2}u_i[k]\\
q_i[k+1]&=q_i[k]+Tu_i[k].
\end{split}
\end{align}
Interestingly, the form of the transformed dynamic system \eqref{eq.mdyn} is the same as \eqref{eq.dyndiscrete}. \textcolor{black}{Based on Theorem 4.2 in \cite{dibajiishiiSCL2015}, the system dynamic \eqref{eq.dyndiscrete} asymptotically reaches consensus (in the sense that $\lim_{k \to \infty} x_j[k]-x_i[k] = \delta_{ij}$) in a network with communication delays and asynchrony with the control input:}

\begin{equation*}
\textcolor{black}{u_i[k]
= \sum_{j \in \mathcal{N}_i} 
a_{ij}[k]\bigl(
x_j[k-\tau_{ij}[k]]-x_i[k] - \delta_{ij}
\bigr) -\alpha_i  v_i[k].}
\end{equation*}

  \textcolor{black}{Equivalently, the transformed dynamic system \eqref{eq.mdyn} asymptotically reaches consensus (in the sense that $\lim_{k \to \infty} p_j[k]-p_i[k] = \delta_{ij}$, $q_i[k]\rightarrow 0$ as $k\rightarrow\infty$, $\forall i,j\in \mathcal{V}$) in a network with communication delays and asynchrony with the control input:}

\begin{equation} \label{eq.pqupdate}
u_i[k]
= \sum_{j \in \mathcal{N}_i} 
a_{ij}[k]\bigl(
p_j[k-\tau_{ij}[k]]-p_i[k] - \delta_{ij}
\bigr) -\alpha_i  q_i[k],
\end{equation}
where $a_{ij}[k]$ is the $(i, j)$ entry of the adjacency matrix $A[k] \in \mathbb{R}^{n\times n}$ associated with $\mathcal{G}$ , $\tau_{ij}[k] \in \mathbb{Z}_+$ denotes the time delay corresponding to the edge $(j,i)$ at time $k$ and $\alpha_i$ is a positive scalar. For the sake of simplicity, $p_j[k-\tau_{ij}[k]]-p_i[k]$ is called the relative position of vehicle $j$ to vehicle $i$ in the rest of the paper. Recalling the variable change of \eqref{eq.varch}, from the viewpoint of vehicle $i$, the most recent information regarding vehicle $j$ at time $k$ is the position of $j$ at time $k-\tau_{ij}[k]$ relative to its own current position. While the communications delays are assumed to have the common upper bound $\tau$, they can be different at each edge and even time varying defined as
\begin{equation}\label{eqn:delay_tau}
0 \leq \tau_{ij}[k] \leq \tau,~~
(j,i)\in\mathcal{E},~k\in\mathbb{Z}_+.
\end{equation}

\textcolor{black}{Using transformation \eqref{eq.varch}, the expressions of $\lim_{k \to \infty} p_j[k]-p_i[k] = \delta_{ij}$ and $\lim_{k \to \infty} q_i[k] = 0$, can be transformed to the expressions of $\lim_{k \to \infty} x_j[k]-x_i[k] = \delta_{ij}$ and $\lim_{k \to \infty} v_i[k]= r$, $\forall i,j\in \mathcal{V} \setminus \mathcal{M}$ which represent what we called ``resilient coordinated movement'' in Definition \ref{def.velcons}.}

%

According to \eqref{eq.pqupdate} and \eqref{eqn:delay_tau}, note that each normal vehicle receives the position value of its neighbors at least once in $\tau$ time steps, but possibly in an asynchronous manner. Also, vehicle~$i$ uses its own velocity without delay in the update rule. The value of $\tau$ is not required to be known to the vehicles as it is not utilized in the update rule.

We also emphasize that, in fully asynchronous settings, vehicles must also be facilitated with their own clocks \cite{JiahuHirche}. However, we consider the so-called \textit{partially asynchronous} updating setting in this paper. This is a common term in the literature for those update protocols with both delay and different update times \cite{Bertsekas} and in fact contains some level of synchrony meaning that all vehicles use the same clock. Considering delays in communicated data packets to address partial asynchrony has been studied in \cite{Gao,Cheng-Lin,JiahuHirche}.

The malicious vehicles are assumed to be omniscient, i.e. they have full knowledge of the topology, updating times, the transmitted data packets by normal vehicles, and even delays $\tau_{ij}[k]$ for all communications and for $k \geq 0$. The malicious vehicles can take advantage of this knowledge to make deceiving back and forth movements or broadcast faulty data packets to confound and prevent the normal vehicles to reach consensus. However, any misbehavior by malicious vehicles might only make the convergence time longer and cannot affect the main outcome of our method, i.e. prevent the vehicles from reaching consensus.

The ADP-MSR algorithm which is executed by each vehicle at each time step $k$ is outlined in Algorithm~\ref{alg1}. The simplicity of this algorithm is its main feature. Each normal vehicle disregards the misleading information -- extreme values -- received from its neighbors by neglecting the incoming edges from those suspicious neighbors. Then, the remaining edges determine the underlying graph $\mathcal{G}[k]$. \textcolor{black}{In this algorithm, there is no need to know the entire topology of the network; instead, we consider a parameter $f$ to secure the consensus in the worst case scenario, in which $f$ is an upper bound for the number of malicious vehicles. That means the normal vehicles do not need to know $f$ accurately. In fact, each vehicle considers that there are at most $f$ malicious nodes within the neighborhood (or within the entire network). Note that these types of assumptions are common in the literature of robust control.}

\begin{algorithm}[t]
	\SetAlgoLined
	At each step time $k$,\\
	\uIf{vehicle $i$ decides to make an update}{
		\For{$j \in \mathcal{N}_i$}{
			The vehicle $i$ calculates $p_j[k-\tau_{ij}[k]]-p_i[k]-\delta_{ij}$ based on the most recent position values.  
		}
	Vehicle $i$ sorts the calculated values from the largest to the smallest.\\
		\uIf{there are less than $f$ vehicles that $p_j[k-\tau_{ij}[k]]-p_i[k] - \delta_{ij} \geq 0$}{
		The normal vehicle $i$ ignores the incoming edges from those vehicles.\\
		}
		\Else{
		The normal vehicle $i$ ignores the incoming edges of $f$ vehicles, which have the largest relative position values.\\
		}
		\uIf{there are less than $f$ vehicles that $p_j[k-\tau_{ij}[k]]-p_i[k] - \delta_{ij} \leq 0$}{
		The normal vehicle $i$ ignores the incoming edges from those vehicles.
		}
		\Else{
		The normal vehicle $i$ ignores the incoming edges of $f$ vehicles, which have the smallest relative position values.
		}
		Vehicle $i$ applies the control input \eqref{eq.pqupdate} by the substitution $a_{ij}[k] = 0$ for edges $(j,i)$ which are ignored.
	}
	\Else{
		Vehicle $i$ applies the control \eqref{eq.pqupdate} where the position values of its neighbors remain the same as time step $k-1$.\\
	}
	\KwResult{$u_i[k]$	
	}

	\caption{ADP-MSR} 
	\label{alg1}
	
\end{algorithm}

\begin{theorem}\label{th.sufnec}\rm
Under the $f$-total malicious model, the network of vehicles with second-order dynamics utilizing the control input \eqref{eq.pqupdate} and the ADP-MSR algorithm comes to resilient coordinated movement with an exponential convergence rate, if the underlying graph is $(2f+1)$-robust, and if it comes to resilient coordinated movement, the underlying graph is at least $(f+1,f+1)$-robust.
\end{theorem}

\begin{proof} 
(Sufficiency) The proof is similar to what is presented in Theorem 4.2 of \cite{dibaji2017resilient}. The proof there is presented for the position consensus of the original dynamic system \eqref{eq.dyn}. Here, the dynamic system is replaced with \eqref{eq.mdyn} and the result is valid for $p[k]$ and $q[k]$, i.e. \textcolor{black}{$p_j[k]-p_i[k] \rightarrow \delta_{ij}$, $q_i[k]\rightarrow 0$, $\forall i,j\in \mathcal{V} \setminus \mathcal{M}$ as $k \rightarrow \infty$. Equivalently, according to \eqref{eq.varch}, $x_j[k]-x_i[k] \rightarrow \delta_{ij}$, $v_i[k]\rightarrow r$ as $k\rightarrow\infty$, $\forall i,j\in \mathcal{V} \setminus \mathcal{M}$. Thus, all the normal vehicles asymptotically reach a same velocity of $r$ and meet a formation with the distributed inter-vehicular distance of $\delta_{ij}$, which is the resilient coordinated movement.}

(Necessity) We consider the synchronous networks without communication delays as the proof of necessity is also valid for the more general case of partially asynchronous networks with communication delays. Contradiction is used for the proof. 
Suppose that the network is not $(f+1,f+1)$-robust. 
Then, there are nonempty disjoint sets $\mathcal{V}_1 , \mathcal{V}_2 \subset \mathcal{V}$ such that none of the conditions 1--3 in Definition~\ref{(def:robustgraph)} are held, i.e.

\begin{itemize}
\item[1.] \textcolor{black}{$\vert \mathcal{X}_{\mathcal{V}_1}^{f+1} \vert < \vert\mathcal{V}_1 \vert$,}
\item[2.] \textcolor{black}{$\vert \mathcal{X}_{\mathcal{V}_2}^{f+1} \vert < \vert\mathcal{V}_2 \vert$,}
\item[3.] \textcolor{black}{$\vert \mathcal{X}_{\mathcal{V}_1}^{f+1} \vert +\vert\mathcal{X}_{\mathcal{V}_2}^{f+1}\vert \leq f$,}
\end{itemize}
\textcolor{black}{where by the definition $\mathcal{X}_{\mathcal{V}_1}^{f+1}$ and $\mathcal{X}_{\mathcal{V}_2}^{f+1}$ are the subsets of $\mathcal{V}_1$ and $\mathcal{V}_2$ whose their nodes have at least $f+1$ incoming links.
We assume that $p_i[0]=a$, $\forall i \in \mathcal{V}_1$ and $p_j[0]=b$, $\forall j \in \mathcal{V}_2$, where $a < b$. Let $p_\ell [0]=c$, where $a \leq c \leq b$, $\forall \ell \in \mathcal{V} \setminus (\mathcal{V}_1 \cup \mathcal{V}_2)$. We also assume that $q_i[0]=0$, $\forall i \in \mathcal{V}$.  
From condition~3, we have $\abs{\mathcal{X}_{\mathcal{V}_1}^{f+1}}+\abs{\mathcal{X}_{\mathcal{V}_2}^{f+1}}\leq f$. Also, we suppose that all vehicles in $\mathcal{X}_{\mathcal{V}_1}^{f+1}$ and $\mathcal{X}_{\mathcal{V}_2}^{f+1}$ are malicious and hold on to the constant values. Therefore, since the condition 1 and 2 are not held, i.e. $\vert \mathcal{X}_{\mathcal{V}_1}^{f+1} \vert < \vert\mathcal{V}_1 \vert$ and $\vert \mathcal{X}_{\mathcal{V}_2}^{f+1} \vert < \vert\mathcal{V}_2 \vert$, we can conclude that there exist at least one normal vehicle in $\mathcal{V}_1$ and one normal vehicle in $\mathcal{V}_2$ which have $f$ or fewer incoming neighbors outside of their own sets as they are not in $\mathcal{X}_{\mathcal{V}_1}^{f+1}$ and $\mathcal{X}_{\mathcal{V}_2}^{f+1}$. Consequently, these normal vehicles in $\mathcal{V}_1$ and $\mathcal{V}_2$ update based only on the values inside $\mathcal{V}_1$ and $\mathcal{V}_2$ by removing the values received from outside of their sets. This makes their values unchanged at $a$ and $b$. Thus, the normal vehicles do not come into agreement.} 
\end{proof}

\textcolor{black}{In our problem setup, we aim to find topological conditions that for \textit{any} set of malicious nodes and \textit{any} malicious behavior, resilient coordinated movement is achieved. For example, what we proved in the necessary condition is that for any topological condition less restrictive than a $(f+1,f+1)$-robust, there is a set of malicious nodes or behaviors that fail the normal vehicles to achieve the resilient coordinated movement. To elaborate more on this, consider a spanning tree which is the necessary condition for the resilient consensus if the malicious nodes take no misleading behavior regardless of the delayed and asynchronous communications. However, the malicious vehicles will have more freedom to deceive more normal vehicles in a network with communication delays and asynchrony. So, the necessary condition for the synchronous case ($(f +1, f +1)$-robustness) is also a
necessary condition for the asynchronous case.}

\textcolor{black}{Also, note that authors in \cite{dibaji2017resilient} have considered the case where all normal agents agree on their positions and then stop. This paper pushes the previous research one step forward in the sense that the positions and velocities of the networked agents reach consensus simultaneously -- the problem whose solution was non-trivial after \cite{dibaji2017resilient}.}

Furthermore, in Theorem~\ref{th.sufnec}, we observe a gap between the sufficiency and the necessity conditions. This point is illustrated by a $2f$-robust graph in Fig.~\ref{fig: AsyncProp}, which is not resilient to $f$ totally bounded adversarial vehicles as we will discuss in what follows. 

\begin{figure}[t]
	\def \svgscale{.45}
	\hspace{.8cm} 
	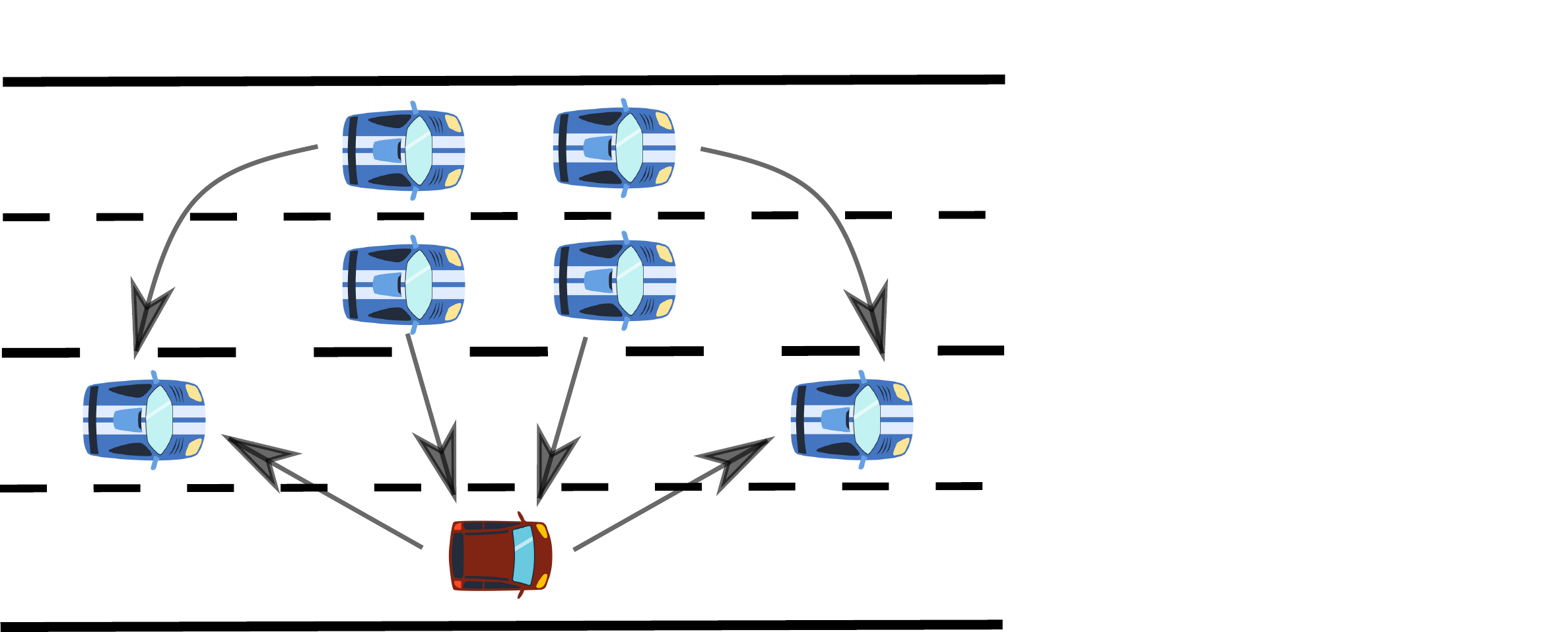	
    \caption{A $2f$-robust network in which vehicles fail to reach coordinated movement with partial asynchrony and delayed information (each vehicle in this figure is representative for a set of $f$ vehicles that are strongly connected). Note that the focus of this paper is on the longitudinal motion of vehicles.} \label{fig: AsyncProp} 
\end{figure}

This graph is composed of four subgraphs $\mathcal{G}_i$, $i=1,\ldots,4$, and each of them is a complete graph -- each vehicle in this figure is representative for a set of $f$ vehicles that are strongly connected. The graph $\mathcal{G}_1$ consists of $4f$ vehicles and the rest have $f$ vehicles. Each vehicle in $\mathcal{G}_4$ has incoming links from $2f$ vehicles of $\mathcal{G}_1$. 
Every vehicle in $\mathcal{G}_3$ has $f$ links from $\mathcal{G}_1$ and $f$ links from $\mathcal{G}_4$. Likewise, each vehicle of $\mathcal{G}_2$ has $f$ link from  $\mathcal{G}_1$ and $f$ incoming links from $\mathcal{G}_4$.

Note that the minimum degree for a $2f$-robust graph is $2f$. However for this graph, the minimum degree of 
the vehicles is $2f+1$ or greater. This is an important point for the following reason. 
If a normal vehicle has only $2f$ neighbors, it will keep its current state since it might ignore all of the values received from its neighbors under the ADP-MSR algorithm. It is clear that coordinated movement cannot take place if this happens for more than two vehicles in the network. The following proposition formally states our claim.
\begin{proposition}\label{Prop: Asynch2frobust}\rm
There exists a 2$f$-robust network with the minimum degree $2f+1$ under which normal vehicles might not reach resilient consensus by the ADP-MSR algorithm. 
\end{proposition}

\begin{proof}
We suppose that the network in Fig.~\ref{fig: AsyncProp} is both $2f$-robust and $(f+1,f+1)$-robust, but resilient coordinated movement cannot be achieved under the ADP-MSR. 
Assume that all vehicles in $\mathcal{G}_2$ are malicious. 
We show a scenario in which 
by the ADP-MSR algorithm, the position values of the vehicles in $\mathcal{G}_3$ and $\mathcal{G}_4$ never concur.

Note that $\mathcal{G}_1$ is 2$f$-robust because of Lemma \ref{lemma:robust graphs} (v) (see \ref{sect: appendix}). By (iv) of this lemma, the graph obtained by adding $\mathcal{G}_2$ is still 2$f$-robust, since there are 2$f$ edges from $\mathcal{G}_1$. 
Similarly, adding $\mathcal{G}_3$ and $\mathcal{G}_4$ and the 
required edges based on (iv) of Lemma \ref{lemma:robust graphs} 
also keeps the graph to be 2$f$-robust. 

We assume vehicles have the following initial states for $k=0$ and the prior $\tau$ steps:
\begin{align*}
&p_i[0]=a-\frac{\delta_{i\ell}}{2}, q_i[0]=0, \forall i \in \mathcal{V}_3, \\
&p_\ell[0]=b+\frac{\delta_{i\ell}}{2}, q_\ell[0]=0, \forall \ell \in \mathcal{V}_2,\\
&p_h[0]=c, q_h[0]=0, \forall h \in \mathcal{V}_1,
\end{align*}
where $a-\frac{\delta_{i\ell}}{2} < c < b+\frac{\delta_{i\ell}}{2}$. 

Also, the malicious vehicles set the following values as their states:
\begin{align*}
p_j[2m]=a-\frac{\delta_{i\ell}}{2}, p_j[2m+1]=b+\frac{\delta_{i\ell}}{2}, \forall j \in \mathcal{V}_4
\end{align*}
and the time delays are chosen by the following scenario: 
\begin{align*}
&\tau_{ij}[2m+1]=1, \forall j \in \mathcal{V}_2, i \in \mathcal{V}_3, (j,i) \in \mathcal{E}, \\
&\tau_{\ell j}[2m+1]=0, \forall j \in \mathcal{V}_2, \ell \in \mathcal{V}_2, (j,\ell) \in \mathcal{E}, \\
&\tau_{ij}[2m]=0,\\
&\tau_{\ell j}[2m]=1,
\end{align*}
where $m \in \mathbb{Z}_+$. All other links have no delay. 
Then, to the vehicles in $\mathcal{G}_3$,
the malicious vehicles appear to be stationary at the state value $a-\frac{\delta_{i\ell}}{2}$ and to the vehicles in $\mathcal{G}_4$ at the state value $b+\frac{\delta_{i\ell}}{2}$.

By executing the ADP-MSR at $k=0$, the vehicles in $\mathcal{G}_3$ will remove the position values of all neighbors 
in $\mathcal{G}_1$ since $a-\frac{\delta_{i\ell}}{2} < c$. 
Thus, for $i\in\mathcal{V}_3$, $p_i[1]=a-\frac{\delta_{i\ell}}{2}$. 
At $k=1$, the same happens for the vehicles $\ell \in \mathcal{V}_4$ 
and they stay at $b+\frac{\delta_{i\ell}}{2}$. Since the vehicles in $\mathcal{G}_3$ are not affected by any vehicles with state values larger than $a-\frac{\delta_{i\ell}}{2}$, they remain at their state value for all future steps. The same holds among the normal vehicles in the network, therefore, $p_i[k]=a-\frac{\delta_{i\ell}}{2}$ and $p_\ell[k]=b+\frac{\delta_{i\ell}}{2}$ for all $i \in \mathcal{V}_3$ and $\ell \in \mathcal{V}_4$. This shows failure in agreement of all vehicles as $\lim_{k \rightarrow \infty } p_\ell[k]-p_i[k]=\delta_{i\ell}+(b-a) \neq \delta_{i\ell}$.
\end{proof}

\subsection{Further Discussions and Results}

Here, we provide some extensions to the discussions and results so far proposed in the paper. 

First, note that only the number of malicious vehicles in each normal vehicle's neighborhood plays a role in the proof of Theorem~\ref{th.sufnec}. Therefore, the result is also valid for the $f$-local malicious model leading to the following corollary.

\begin{corollary}\label{Corollory: f-localAsycnronousSecond-Order}\rm
Consider the network of vehicles with second-order dynamics using the control input proposed in \eqref{eq.pqupdate} and the ADP-MSR algorithm. The network achieves resilient coordinated movement under the $f$-local malicious model if the underlying graph is $(2f+1)$-robust. 
\end{corollary} 
Note that the results in this paper are all valid for the second-order networks whose underlying graphs are fixed, i.e. with time-invariant $\mathcal{E}$. In \cite{LeBlancPaper}, for the first-order synchronous vehicle networks, there is a natural extension for the time-varying $\mathcal{G}[k]$ and based on that $\mathcal{G}[k]=(\mathcal{V},\mathcal{E}[k])$ is enough to be $(f+1,f+1)$-robust at each time $k$. The same condition is valid here, again for second-order synchronous networks. However, the assumption on robustness of the graph \textit{at each time $k$} is quite conservative and might bring difficulties in practice. Here, we would like to state a new relaxed condition for the partially asynchronous time-varying networks. The following definition has a key role for this purpose:
\begin{definition}  \rm
(Jointly $r$-robust) The time-varying graph $\mathcal{G}[k]=(\mathcal{V},\mathcal{E}[k])$ is jointly $r$-robust if there exists a fixed $\ell$ such that the union of $\mathcal{G}[k]$ over each consecutive $\ell$ steps is $r$-robust.
\end{definition}

In a time-varying network, each normal vehicle $i$ can use the outdated links from $\tau$ time steps back whenever some information is not available. Thus, the sufficient condition is obtained with the following additional assumptions:
\begin{equation}\label{As: assumption}
\ell \leq \tau.
\end{equation} 
By the above discussions, the sufficient condition is presented as below.
\begin{corollary}\label{Cor: Time-varying}\rm
Under the $f$-total/$f$-local malicious model, 
the time-varying network of vehicles with second-order dynamics utilizing control input \eqref{eq.pqupdate} and ADP-MSR algorithm achieves resilient coordinated movement, if the underlying graph is jointly $(2f+1)$-robust under condition \eqref{As: assumption}.	
\end{corollary}

\begin{figure}[t]
	\def \svgscale{.45}
	\hspace{.8cm} 
	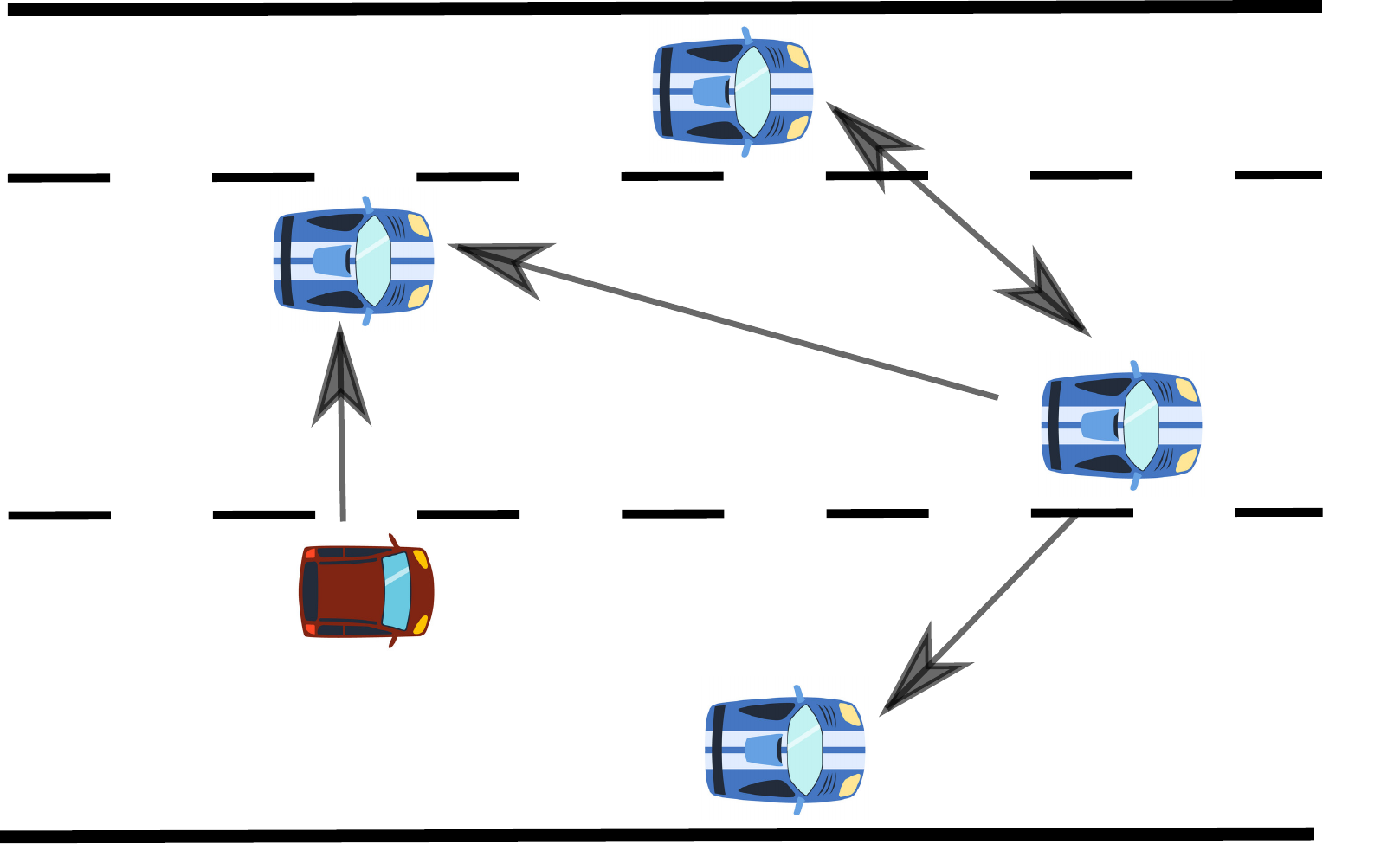
    \caption{A graph that is  $(2,2)$-robust but not $3$-robust.} \label{fig:robustgraph} 
\end{figure}

Now, we discuss the relation between the graph properties proposed here for the resilient coordinated movement problem and those used in standard consensus problems without any malicious agents \cite{CaoMorseAnderson,JiahuQin2012,JiahuHirche,Xiao2006}.
In this paper, it is assumed that the number of 
adversarial vehicles is upper bounded by $f$. 
By removing all edges connected
to malicious vehicles, we can obtain a subgraph of $\mathcal{G}$ consisting of only the normal vehicles. 
By (vi) of Lemma~\ref{lemma:robust graphs} (see \ref{sect: appendix}), 
this network becomes $(1,f+1)$-robust. 
Now, the obtained graph has a spanning tree according to (iv) of 
the same lemma. We know that consensus can be achieved under such a graph. It is also interesting that the sufficient condition in Corollary \ref{Cor: Time-varying} is consistent with the consensus condition on time-varying networks known as having jointly spanning tree \cite{Boyd2005}.

It is further noted that \cite{Azadmanesh2002,
khanafer,Azadmanesh1993,LeBlancPaper,plunkett,Vaidya} consider the so-called omissive faults, where malicious nodes \textit{can deny} making any transmissions. Therefore, the normal vehicle $i$ would wait to receive the position values of at least $d_i-f$ neighbors before making an update.
It should be noted that omissive faults can also be tolerated 
by the MSR-type algorithms. The malicious vehicles knowing that 
the normal vehicles apply the ADP-MSR algorithm
might attempt to make this kind of attack to 
cause denial of information for filtering the received values in Algorithm~\ref{alg1}. In such cases, if vehicle $i$ does not receive the data packets from $m_i[k]$ incoming neighbors at time $k$, then the parameter of the ADP-MSR for that vehicle can be changed from $2f$ to $2(f-m_i[k])$ assuming that vehicle $i$ is aware of $d_i[k]$. The topology analysis remains mostly the same.

Besides, the ADP-MSR algorithms for $f$-total malicious models are resilient against another type of adversaries 
studied in \cite{Feng}. There, the adversarial agents can extend the network by adding extra links. However, this does not change the value of $f$ in the network. For example, in our problem, this can happen in a highway when some additional vehicles are passing by the connected vehicles network. However, note that the situation is subtly different in the case with the $f$-local model. Adding an extra link might increase the number of malicious vehicles in a neighborhood of some normal vehicles. Accordingly, the vehicles must know which links are newly created so as to remove them along with the edges ignored in the ADP-MSR algorithm.

\section{Numerical Example}\label{sect: simulations}

\begin{figure}[t] 
	\begin{center}
	\includegraphics[scale=0.55]{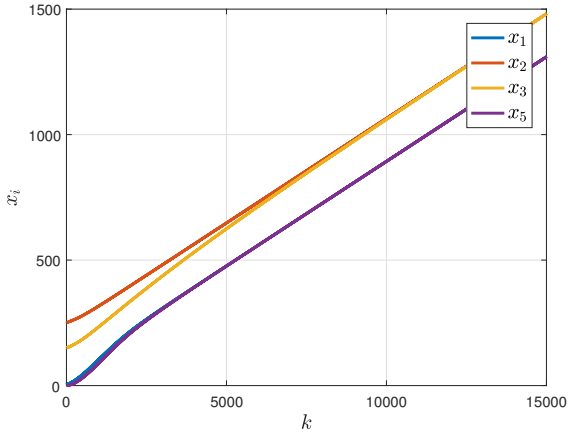}
	\caption{Position-time history of vehicles under $(2,2)$-robust graph - coordinated movement failed.}
	\label{fig.xfail}
	\vspace{-4mm}
	\end{center}
\end{figure}

\begin{figure}[t] 
	\begin{center}
	\includegraphics[scale=0.55]{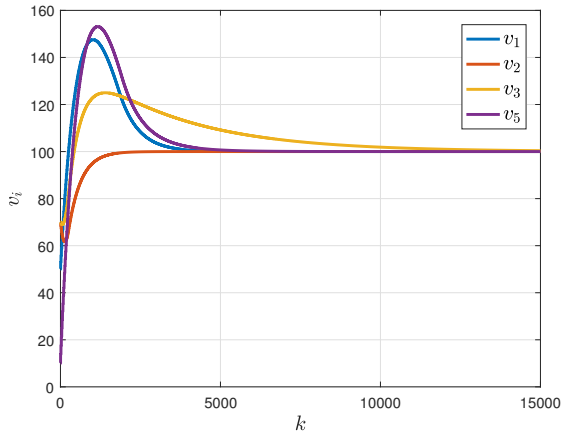}
	\caption{Velocity-time history of vehicles under $(2,2)$-robust graph - coordinated movement failed.}
	\label{fig.vfail}
	\vspace{-4mm}
	\end{center}
\end{figure}

\begin{figure}[t] 
	\begin{center}
	\includegraphics[scale=0.55]{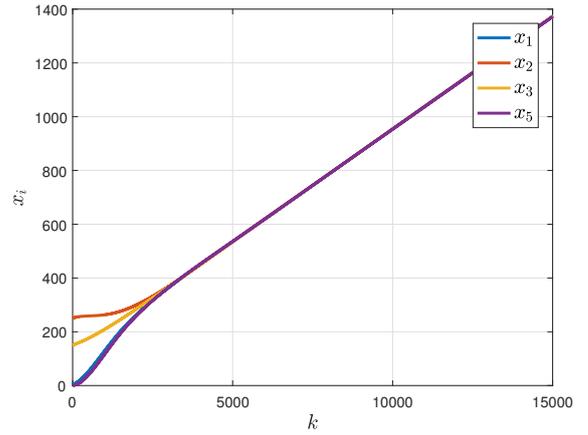}
	\caption{Position-time history of vehicles under $3$-robust graph - coordinated movement succeeded.}
	\label{fig.xagree}
	\vspace{-4mm}
	\end{center}
\end{figure}

\begin{figure}[t] 
	\begin{center}
	\includegraphics[scale=0.55]{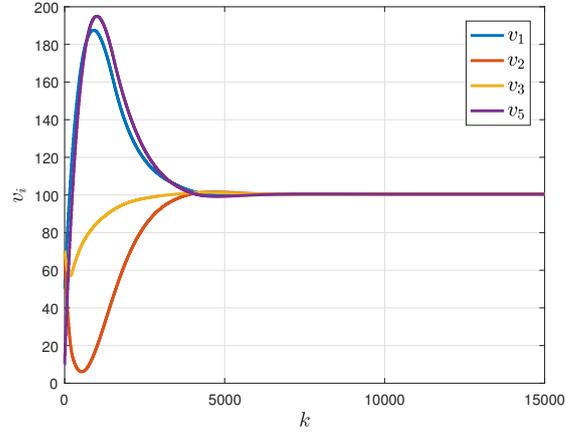}
	\caption{Velocity-time history of vehicles under $3$-robust graph - coordinated movement succeeded.}
	\label{fig.vagree}
	\vspace{-4mm}
	\end{center}
\end{figure}

Suppose a network of vehicles connected together on the network illustrated in Fig.~\ref{fig:robustgraph} with partially asynchronous delayed settings. This graph is $(2,2)$-robust (refer to \ref{sect: appendix} for more discussion). This network is considered to contain only one malicious vehicle, i.e. $f=1$. The sampling period is set to be $T=0.01$. Note that we set $\delta_{ij} = 0$ in this example, thus the vehicles come to consensus in parallel lanes with no relative longitudinal distances. Also, remember that we only consider longitudinal motion of the vehicles. Therefore, in this example, we assume that they move in parallel or each vehicle utilizes a CAS to overtake frontier vehicles if required. We consider two different scenarios to show the effectiveness of our method in the presence of active and passive sensing systems. In both scenarios, four normal vehicles periodically make updates within each 12 time steps with various timings. Specifically, vehicles 1, 2, 3, and 5 make updates at time steps $k=12\ell+6,12\ell+9,12\ell+11,12\ell+4$ for $\ell\in\mathbb{Z}_{+}$, respectively. We assume that at these time steps, their updates are made without any delays. However, each vehicle deals with nonuniform time-varying delays ($\tau=11$) since the normal vehicles do not receive new information at other time steps.

\textbf{Setting 1:} All the normal vehicles are assumed to have passive sensing systems, e.g. GPS receiver, for navigation. Thus, the malicious vehicle can misbehave them by easily sending incorrect information to them instead of its actual position and velocity, and is free to move in its own way. It can be even stopped somewhere on the road and broadcast its false information.

To simulate this scenario, the initial states of the vehicles are given by $\left[x^T[0]~v^T[0]\right]=\big[4~250~150~8~0~50~70~70~60~10\big]$. The parameters $\alpha_i$ in \eqref{eq.pqupdate} are evaluated as $\alpha_1 = \alpha_5 = 2$ and $\alpha_2 = \alpha_3 = 3$. The desired target velocity of the network of vehicles is set as $r=100$. In this network, the malicious vehicle~4 misguide the normal ones and divide them into multiple groups to prevent them from coming to a single agreement. To this end, vehicle~4 incorrectly send its positions as: $x_4[2k]=2+kTr$ and $x_4[2k+1]=200$ for all $k \geq 0$. Figs.~\ref{fig.xfail} and \ref{fig.vfail} illustrate the time history of the positions and velocities of the normal vehicles.
As expected, the positions of the normal vehicles do not reach consensus although the underlying network is $(2,2)$-robust, as a necessary condition. In fact, the ADP-MSR cannot stop the malicious vehicle from misguiding the normal vehicles. Fig.~\ref{fig.xfail} indicates that in fact vehicles are divided into two groups and move with a relative distance because of the malicious behavior of vehicle~4, sending the false data to the normal vehicles.
 
Next, we obtain a 3-robust graph (which is $2f+1$-robust in this case) by adding enough edges. As illustrated in Fig.~\ref{fig.xagree} and \ref{fig.vagree}, the same simulation with the complete graph with 5 nodes (the only $3$-robust graph with 5 nodes) verifies the sufficient condition of Theorem~\ref{th.sufnec} for the partially asynchronous setting.

\begin{figure}[t] 
	\begin{center}
	\includegraphics[scale=0.55]{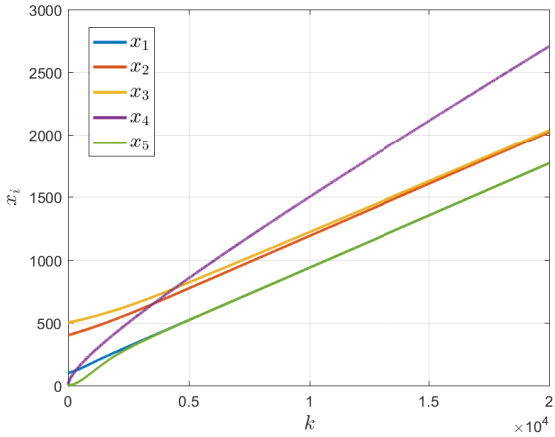}
	\caption{Position-time history of vehicles under $(2,2)$-robust graph - coordinated movement failed even when vehicles 2 and 3 have active sensing systems.}
	\label{fig.xfail.active}
	\vspace{-4mm}
	\end{center}
\end{figure}

\textbf{Setting 2:} Normal vehicles 2 and 3 are assumed to have active sensing systems and vehicles 1 and 5 are assumed to have passive sensing systems for navigation. As a result, the malicious vehicle faces a more challenging situation rather than the first setting\footnote{Note that more complicated situations in which actions and reactions of the normal and malicious vehicles are taken into account are out of scope of this paper. Game theoretic approaches will address them.}. If the malicious vehicle wants to misguide vehicles 2 and 3, it has to appear around them. However, as vehicles 1 and 5 each has a passive sensing system, the malicious vehicle can still use broadcasting false information to avoid them reaching consensus with the other two vehicles. As shown in Fig.~\ref{fig.xfail.active}, the malicious vehicle can affect the consensus and misguide the vehicles into two groups even when 2 of the vehicles have active sensing systems. On the other hand, as shown in Fig.~\ref{fig.xagree.active}, the sufficient graph condition, $(2f+1)$-robustness, guarantee the longitudinal coordinated movement of the vehicles. It is clear that the cyber attack for the malicious vehicle could be more difficult (yet possible) if all the vehicles are equipped with active sensing systems. Theoretically, the malicious vehicle cannot prevent consensus even if all the vehicles are equipped with active sensing systems. However, finding a practical scenario for this case would be tricky.

In this setting, the initial states of the vehicles are given by $\left[x^T[0]~v^T[0]\right]=\big[100~400~500~10~0~50~70~70~60~10\big]$. The parameters $\alpha_i$ in \eqref{eq.pqupdate} are evaluated as $\alpha_1 = \alpha_5 = 2$ and $\alpha_2 = \alpha_3 = 10$. Also, vehicle~4 incorrectly send its positions to vehicle 1 and 5 as: $x_4[2k+1]=200$ for all $k \geq 0$, while it moves closer to vehicle 2 and 3 most of the time and appears as: $x_4[k] = 0.1k+5 \sqrt{k}$.

\begin{figure}[t] 
	\begin{center}
	\includegraphics[scale=0.55]{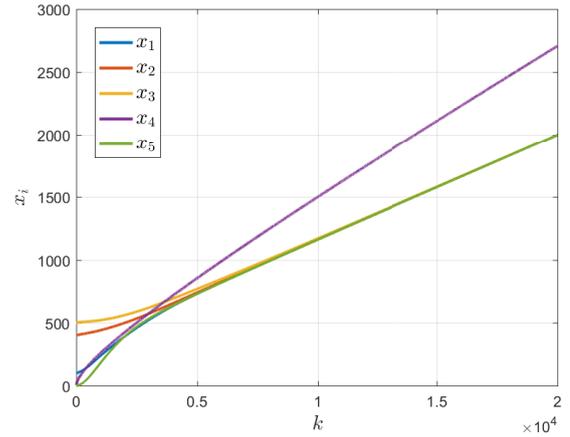}
	\caption{Position-time history of vehicles under $3$-robust graph - coordinated movement succeeded. Vehicles 2 and 3 have active sensing systems.}
	\label{fig.xagree.active}
	\vspace{-4mm}
	\end{center}
\end{figure}

\section{Conclusion}\label{sect: conclusion}
In this paper, we studied the problem of resilient coordinated movement of a network of connected vehicles with second-order longitudinal dynamics, where the number of malicious vehicles in the network is bounded by a parameter $f$, known to the vehicles. We have proposed a distributive strategy for the normal vehicles 
to achieve resilient consensus on their positions with a safe inter-vehicular distance and a predefined target velocity. The necessary and sufficient graph conditions are respectively $(f+1,f+1)$-robustness and $(2f+1)$-robustness for resilient coordinated movement of the network of vehicles under the $f$-total malicious model. Each vehicle performs the proposed update rule and ADP-MSR algorithm to achieve the agreement with an exponential convergence rate. Communications in the network are partially asynchronous with bounded delays.

Future research trend can be possibly investigations in the following two main directions: i) considering 2D coordinated movement of connected vehicles and develop appropriate updating algorithms, ii) finding a necessary and sufficient topology condition for the convergence.

\appendix
\section{} \label{sect: appendix}
To have a better understanding of $(r,s)$-robust graphs \cite{LeBlancPhD}, the following lemma is presented.

\begin{lemma}\label{lemma:robust graphs}\rm
For an $(r,s)$-robust graph $\mathcal{G}$, the followings hold:
\begin{enumerate}
\item[(i)] $\mathcal{G}$ is $(r',s')$-robust, where $0\leq r'\leq r$ and $1 \leq s'\leq s$, and in particular, it is $r$-robust.
\item[(ii)] $\mathcal{G}$ is $(r-1,s+1)$-robust.
\item[(iii)] $\mathcal{G}$ is at least $r$-connected, but an $r$-connected graph is not necessarily $r$-robust. 
\item[(iv)] $\mathcal{G}$ has a directed spanning tree.
\item[(v)] $r \leq \lceil n/2 \rceil$. Also, if $\mathcal{G}$ is a complete graph, 
then it is $(r',s)$-robust for all $0<r'\leq \lceil n/2 \rceil$ and $1 \leq s \leq n$.
\item[(vi)] The graph $\mathcal{G}'=( \mathcal{V} ,\mathcal{E}_{0} )$ is 
$(r-w,s)$-robust, when $\mathcal{G}'$ is formed by removing 
at most $w$ edges from neighbors of each node in $\mathcal{V}$, where $w < r$. 
\item[(vii)] The graph $\mathcal{G}'=( \mathcal{V} \cup  \{ v_{0} \}, \mathcal{E} 
\cup {\mathcal{E} _{0}} )$, where $v_{0}$ is a node added to 
$\mathcal{G}$ and $\mathcal{E} _{0}$ is the edge set related to 
$v_{0}$, is $r$-robust if ${d}_{v_{0}} \geq r+s-1$.  
\end{enumerate}
Moreover, a graph is $(r,s)$-robust 
if it is $(r+s-1)$-robust.
\end{lemma}

Generally, it is clear that $(r,s)$-robustness is more restrictive than $r$-robustness. The 5 nodes graph in Fig.~\ref{fig:robustgraph}  can be shown to be $(2,2)$-robust, but not $3$-robust. 
From computational point of view, checking robustness properties is difficult since the problem needs combinatorial calculations. However, tending the size of random graphs to infinity makes them robust \cite{zhangfata}.


\bibliography{mybibfile}

\end{document}

%% file: hw.pdf_tex
\begingroup%
  \makeatletter%
  \providecommand\color[2][]{%
    \errmessage{(Inkscape) Color is used for the text in Inkscape, but the package 'color.sty' is not loaded}%
    \renewcommand\color[2][]{}%
  }%
  \providecommand\transparent[1]{%
    \errmessage{(Inkscape) Transparency is used (non-zero) for the text in Inkscape, but the package 'transparent.sty' is not loaded}%
    \renewcommand\transparent[1]{}%
  }%
  \providecommand\rotatebox[2]{#2}%
  \newcommand*\fsize{\dimexpr\f@size pt\relax}%
  \newcommand*\lineheight[1]{\fontsize{\fsize}{#1\fsize}\selectfont}%
  \ifx\svgwidth\undefined%
    \setlength{\unitlength}{683.91813023bp}%
    \ifx\svgscale\undefined%
      \relax%
    \else%
      \setlength{\unitlength}{\unitlength * \real{\svgscale}}%
    \fi%
  \else%
    \setlength{\unitlength}{\svgwidth}%
  \fi%
  \global\let\svgwidth\undefined%
  \global\let\svgscale\undefined%
  \makeatother%
  \begin{picture}(1,0.40296364)%
    \lineheight{1}%
    \setlength\tabcolsep{0pt}%
    \put(0,0){\includegraphics[width=\unitlength,page=1]{hw.pdf}}%
    \put(0.31211723,0.36920823){\color[rgb]{0,0,0}\makebox(0,0)[lt]{\lineheight{0.625}\smash{\begin{tabular}[t]{l}$\mathcal{G}_1$\end{tabular}}}}%
    \put(0.606244,0.122369){\color[rgb]{0,0,0}\makebox(0,0)[lt]{\lineheight{0.62500006}\smash{\begin{tabular}[t]{l}$\mathcal{G}_3$\end{tabular}}}}%
    \put(0.00502481,0.12067067){\color[rgb]{0,0,0}\makebox(0,0)[lt]{\lineheight{0.62500006}\smash{\begin{tabular}[t]{l}$\mathcal{G}_2$\end{tabular}}}}%
    \put(0.3674166,0.0221657){\color[rgb]{0,0,0}\makebox(0,0)[lt]{\lineheight{0.62500006}\smash{\begin{tabular}[t]{l}$\mathcal{G}_4$\end{tabular}}}}%
    \put(0,0){\includegraphics[width=\unitlength,page=2]{hw.pdf}}%
  \end{picture}%
\endgroup%

%% file: hwsim.pdf_tex
\begingroup%
  \makeatletter%
  \providecommand\color[2][]{%
    \errmessage{(Inkscape) Color is used for the text in Inkscape, but the package 'color.sty' is not loaded}%
    \renewcommand\color[2][]{}%
  }%
  \providecommand\transparent[1]{%
    \errmessage{(Inkscape) Transparency is used (non-zero) for the text in Inkscape, but the package 'transparent.sty' is not loaded}%
    \renewcommand\transparent[1]{}%
  }%
  \providecommand\rotatebox[2]{#2}%
  \newcommand*\fsize{\dimexpr\f@size pt\relax}%
  \newcommand*\lineheight[1]{\fontsize{\fsize}{#1\fsize}\selectfont}%
  \ifx\svgwidth\undefined%
    \setlength{\unitlength}{465.64941887bp}%
    \ifx\svgscale\undefined%
      \relax%
    \else%
      \setlength{\unitlength}{\unitlength * \real{\svgscale}}%
    \fi%
  \else%
    \setlength{\unitlength}{\svgwidth}%
  \fi%
  \global\let\svgwidth\undefined%
  \global\let\svgscale\undefined%
  \makeatother%
  \begin{picture}(1,0.60254203)%
    \lineheight{1}%
    \setlength\tabcolsep{0pt}%
    \put(0,0){\includegraphics[width=\unitlength,page=1]{hwsim.pdf}}%
    \put(0.60186277,0.53549789){\color[rgb]{0,0,0}\makebox(0,0)[lt]{\lineheight{0.625}\smash{\begin{tabular}[t]{l}$1$\end{tabular}}}}%
    \put(0,0){\includegraphics[width=\unitlength,page=2]{hwsim.pdf}}%
    \put(0.87502345,0.27594466){\color[rgb]{0,0,0}\makebox(0,0)[lt]{\lineheight{0.62500006}\smash{\begin{tabular}[t]{l}$2$\end{tabular}}}}%
    \put(0.64834227,0.04572835){\color[rgb]{0,0,0}\makebox(0,0)[lt]{\lineheight{0.62500006}\smash{\begin{tabular}[t]{l}$3$\end{tabular}}}}%
    \put(0.15701827,0.15901943){\color[rgb]{0,0,0}\makebox(0,0)[lt]{\lineheight{0.62500006}\smash{\begin{tabular}[t]{l}$4$\end{tabular}}}}%
    \put(0.14745342,0.39817268){\color[rgb]{0,0,0}\makebox(0,0)[lt]{\lineheight{0.62500006}\smash{\begin{tabular}[t]{l}$5$\end{tabular}}}}%
    \put(0,0){\includegraphics[width=\unitlength,page=3]{hwsim.pdf}}%
  \end{picture}%
\endgroup%